\documentclass[smallabstract,smallcaptions]{dccpaper}

\usepackage{epsfig}
\usepackage{citesort}
\usepackage{amsmath}
\usepackage{amssymb}
\usepackage{color}
\usepackage{url}

\usepackage{cite}

\usepackage{enumitem}

\newlength{\figurewidth}
\newlength{\smallfigurewidth}

\newtheorem{definition}{Definition}
\newtheorem{theorem}{Theorem}
\newtheorem{lemma}{Lemma}

\newtheorem{corollary}{Corollary}

\newenvironment{proof}{\trivlist\item[]\emph{Proof}:}%
{\unskip\nobreak\hskip 1em plus 1fil\nobreak$\Box$
	\parfillskip=0pt%
	\endtrivlist}

\def\idtt#1{\ensuremath{\mathtt{#1}}}
\def\occ{\idtt{occ}}
\def\hp\_leaf{\idtt{hp\_leaf}}
\def\hproot{\idtt{hp\_root}}
\def\str{\idtt{str}}
\def\partner{\idtt{partner}}
\def\path{\idtt{Path}}
\def\loci{\idtt{loci}}
\def\SpecialLeaves{\idtt{SpecialLeaves}}
\def\SpecialSkylineList{\idtt{SpecialSkylineList}}
\def\Special{\idtt{Special}}
\def\lMost{\idtt{lMost}}
\def\rMost{\idtt{rMost}}
\def\lca{\idtt{lca}}

\def\len{\idtt{len}}
\def\parent{\idtt{parent}}
\def\Null{\idtt{NULL}}
\def\rank{\idtt{rank}}
\def\size{\idtt{size}}
\def\select{\idtt{select}}
\def\sp{\idtt{SP}}
\def\root{\idtt{root}}
\def\desc{\idtt{desc}}

\usepackage{xpatch}
\usepackage{algorithm}
\usepackage[noend]{algorithmic}
\usepackage{caption}
\xpatchcmd{\algorithmic}{\setcounter}{\algorithmicfont\setcounter}{}{}
\providecommand{\algorithmicfont}{}
\providecommand{\setalgorithmicfont}[1]{\renewcommand{\algorithmicfont}{#1}}

\setalgorithmicfont{\footnotesize}

\begin{document}
	
	\title
	{\large
		\textbf{Computing Matching Statistics on Repetitive Texts}
	}

	\author{%
		Younan Gao\\[0.5em]
		{\small\begin{minipage}{\linewidth}\begin{center}
					\begin{tabular}{c}
						Dalhousie University \\
						Halifax, Canada\\
						\url{yn803382@dal.ca}
					\end{tabular}
		\end{center}\end{minipage}}
	}

	\maketitle
	\thispagestyle{empty}

	\begin{abstract}
		Computing the {\em matching statistics} of a string $P[1..m]$ with respect to a text $T[1..n]$ is a fundamental problem which has application to genome sequence comparison.
		In this paper, we study the problem of computing the matching statistics upon highly repetitive texts.
		We design three different data structures that are similar to LZ-compressed indexes. 
		The space costs of all of them can be measured by $\gamma$, the size of the smallest string attractor [STOC'2018] and $\delta$, a better measure of repetitiveness  [LATIN'2020]. 	
	\end{abstract}
	
	\Section{Introduction}
	The {\em matching statistics}, MS, of a pattern $P[1..m]$ with respect to a text $T[1..n]$ is an array of $m$ integers such that the $i$-th entry $\textnormal{MS}[i]$ stores the length of the longest prefix of $P[i..m]$ that occurs in $T$.
	For example, given that $T=$``$aaabbbcc$'' and $P=$``ccabb'', matching statistics MS of $P$ w.r.t.\ $T$ stores an array of 5 integers, $[2, 1, 3, 2, 1]$.
	Originally, the concept, matching statistics, was introduced by Chang and Lawler \cite{chang1994sublinear} to solve the approximate string matching problem, i.e., given text $T[1..n]$ and pattern $P[1..m]$, the problem asks for all the locations in the text where $P[1..m]$ appears, with at most $k$ differences (including substitutions, insertions, and deletions) being allowed, where $k$ is not necessarily a constant. 
	The approximate string matching plays an important role in computational genomics.
	In terms of sequence alignment, reads might not match the genome exactly because of the sequencing error, natural variance (i.e., differences in DNA among individuals of the same species), etc.; for that reason, the algorithms for the exact string matching might not be sufficient, and the approximate string matching is needed.
	Matching statistics is also useful in a variety of other applications such as finding the longest common substring between $P$ and $T$ \cite{gusfield1997algorithms}.
	A textbook solution \cite{gusfield1997algorithms} shows that a {\em suffix tree} data structure augmented with {\em suffix links} on the tree nodes can be used to compute matching statistics in $O(m\lg \sigma)$ time, where $\sigma$ represents the size of the alphabet that $T$ is drawn from;
	the data structure uses $O(n)$ words of space.
	Ohlebusch et al.~\cite{ohlebusch2010computing} solved this problem using a fully compressed text indexes built upon $T$, which consist of a {\em wavelet tree} data structure that supports the LF-Mapping and the backward search, a LCP-array, and a data structure that supports fast-navigating on a LCP-interval tree. 
	Their indexes occupy $n\lg \sigma+ 4n+o(n\lg \sigma)$ bits of space and achieve the same computing time as of the textbook solution.
	In the genomic databases, texts are always massive and highly repetitive; however, the compressed indexes based on statistical entropy might not capture repetitiveness \cite{navarro2007compressed}.
	Bannai et al.~\cite{bannai2020refining} considered to compute MS for a highly repetitive text.
	They augmented a {\em run-length BWT} with $O(r)$ words of space, where $r$ is the number of {\em runs} in the BWT for $T$, and their indexes support computing MS in $O(m\lg \lg n)$ time, assuming that each element in $T$ can be accessed in $O(\lg \lg n)$ time.
	Let $z$ denotes the number of phrases in {\em Lempel-Ziv parsing} (LZ).
	It has been proved that $r=O(z\lg^2 n)$ holds for every text $T$ \cite{kempa2020resolution}.
	However, to our knowledge, LZ-based indexes on computing MS have not been known prior to this work.
	
	Recently, new compressibility measures such as $\gamma$, the size of the smallest {\em string attractor}, and $\delta$, a better measure of repetitiveness have been proposed. 
	Both new measures better capture the compressibility of repetitive strings.
	It has been proved that $\delta\leq \gamma\leq z = O(\delta \lg \frac{n}{\delta})$ \cite{kempa2018roots, kociumaka2021towards}.
	In this paper, we design the first string attractor based indexes (which is also workable upon LZ-parsing) to support computing matching statistics; the space cost of the indexes is measured by $\gamma$ and $\delta$.
	The computation time using string attractor based indexes might not be as efficient as the one using BWT-based indexes, but the indexes in the prior category always have an advantage of space cost.

	\paragraph{Our Results.} The results can be summarized as Theorem \ref{theorem_for_all}.
	In the first solution, we apply a data structure similar to LZ-compressed indexes. Instead of using LZ parsing, we define the phrases upon the smallest string attractor. 
	We store a Patricia tree for the reversed phrases and another one for the suffixes following the phrase boundaries.
	To access text $T[1..n]$ within compressed space, we apply the string indexing data structure by Kociumaka et al. \cite{kociumaka2021towards}, whose space cost is measured by $\delta$.
	We give a simple and practical algorithm that reduces the problem of computing MS into $O(m^2)$ times of 2D orthogonal range predecessor queries upon $\gamma$ points on a grid.
	In the second solution, we apply the data structure framework by Abedin et al. \cite{abedin2018heaviest}. Originally, they used the framework to find the longest common substring (LCS) between $P$ and $T$. Naively, computing MS of $P[1..m]$ can be reduced to $m$ times of LCS queries. Given that each LCS query can be computed in $O(m\lg \gamma \lg \lg \gamma)$ time, the naive method would take $O(m \cdot m\cdot\lg \gamma \lg \lg \gamma)$ time. 
	We adjust their framework to computing MS and improve the query time to be $O(m \cdot (m+\lg \gamma \lg \lg \gamma))$.
	
	\begin{theorem}
		\label{theorem_for_all}
		The matching statistics of a $P[1..m]$ with respect to a text $T[1..n]$ can be computed (\romannumeral1) in $O(m^2 \lg^{\epsilon} \gamma + m\lg n)$ time using an $O(\delta \lg \frac{n}{\delta})$ word space data structure, or (\romannumeral2) in $O(m^2 + m \lg \gamma \lg \lg \gamma + m\lg n)$ time using $O(\gamma \lg \gamma + \delta \lg \frac{n}{\delta})$ word space data structure, in which $\epsilon$ is any small positive constant, $\gamma$ is the size of the smallest string attractor, and $\delta$ is $\max\{ S(k)/ k, 1\leq k\leq n \}$, where  $S(k)$ denotes the number of distinct $k$-length sub-strings of $T$.
	\end{theorem}
	
	If text $T[1..n]$ is drawn from constant-size alphabet, we can further improve the computation time to be $O(m^2 + m\lg n)$ time using an $O(\gamma \lg \gamma + \delta \lg \frac{n}{\delta})$ word data structure.
	The third solution can be achieved by combining the first and second solution:
	i) when $m=\Omega(\lg \gamma \lg \lg \gamma)$, we can directly apply the second solution to achieve the target bound for the query time; ii) otherwise, we update the first solution using the technique solving the {\em ball inheritance} problem \cite{chan2011orthogonal} to improve the query time without decreasing the space cost.
	Due to the space limitation, the proof of the third solution is deferred to the full version of this paper.
	%
	%
	%
	%

	\Section{Preliminaries}
	This section introduces the notations and the previous results used throughout this paper. Let $\epsilon$ denote any small positive constant, and all problems are studied under the standard {\em word RAM} model.
	
	\paragraph{Compressibility Measures.} We give the precise definitions of the compressibility measures $\gamma$ and $\delta$ that are mentioned before.
	
	\begin{definition}\cite{kempa2018roots}
		A string attractor of a string $T[1..n]$ is a set of $\gamma'$ positions $\Gamma'= \{j_1, \cdots, j_{\gamma'}\}$ such that every substring $T[i..j]$ has an occurrence $T[i'..j']=T[i..j]$ with $j_k \in [i', j']$ for some $j_k\in \Gamma'$.
	\end{definition}
	
	Let $\Gamma^{*}$ denote $\{1, \Gamma, n\}$, where $\Gamma$ denotes the smallest size string attractor storing positions sorted increasingly.
	For each $2\leq i\leq |\Gamma^*|$, we call substring $T[\Gamma^{*}[i-1].. \Gamma^{*}[i]]$ a {\em parsing phrase}.
	Let $\gamma$ denote the size of $\Gamma$.
	It follows that given any substring $T[i..j]$, there must be an occurrence $T[i'..j']=T[i..j]$ such that $T[i'..j']$ crosses the phrase boundary.
	Kociumaka et al.~\cite{kociumaka2021towards} defined a new measure $\delta$, which is even smaller than $\gamma$.
	Furthermore, measure $\delta$, different from $\gamma$, can be computed in linear time.
	\begin{definition}\cite{kociumaka2021towards}
		Let $d_k(S)$ be the number of distinct length-$k$ sub-strings in $S$. Then
		$\delta = \max\{d_k(S)/k: k\in[1..n]\}.$
	\end{definition}
	
	
	\paragraph{Induced-Check and Find Partner.} Let $T_{1}$ and $T_{2}$ be two trees on the same set of $n$ leaves.
	A node from $T_{1}$ and a node from $T_{2}$ are {\em induced} together if they have a common leaf descendant \cite{gagie2013heaviest}.
	The {\em partner} operation \cite{abedin2018heaviest} is defined upon the inducing relationship.
	\begin{definition}\cite{abedin2018heaviest} 
		Given a pair of trees $T_{1}$ and $T_{2}$, the partner of a node $x\in T_{1}$ w.r.t a node $y \in T_{2}$, denoted by $\partner(x/y)$, is the lowest ancestor, $y'$, of $y$ such that $x$ and $y'$ are induced. Likewise, $\partner(y/x)$ is the lowest ancestor, $x'$, of $x$ such that $x'$ and $y$ are induced.
	\end{definition}
	
	
	It has been proved that the induced relationship can be solved using 2D orthogonal range emptiness queries, while finding the partner can be reduced to 2D orthogonal range predecessor/successor queries.
	
	
	\begin{lemma}\cite{abedin2018heaviest}
		(Induced-Check). \label{lem:induced-check}
		Given two nodes $x$, $y$, where $x\in T_{1}$ and $y \in T_{2}$, we can check if they are induced or not i) in $O(\lg \lg n)$ time using an $O(n\lg \lg n)$ word space structure, or ii) in $O(\lg^{\epsilon} n)$ time using an $O(n)$ word space structure.
	\end{lemma}
	
	\begin{lemma}\cite{abedin2018heaviest}
		\label{lemma:partner_finding}
		(Find Partner). Given two nodes $x$, $y$ where $x\in T_{1}$ and $y\in T_{2}$, we can find $\partner(x/y)$ as well as $\partner(y/x)$ i) in $O(\lg \lg n)$ time using an $O(n\lg \lg n)$ word space structure, or ii) in $O(\lg^{\epsilon} n)$ time using an $O(n)$ word space structure.
	\end{lemma}
	
	\paragraph{String Indexing.} Recently, Kociumaka et al. \cite{kociumaka2021towards} showed that within $O(\delta \lg \frac{n}{\delta})$  words of space, one can represent and index a string of $n$ characters.
	\begin{lemma}\cite{kociumaka2021towards}
		\label{lemma: self-index}
		Given a string $T[1..n]$ with measure $\delta$, one can build an $O(\delta \lg \frac{n}{\delta})$ word space data structure in $O(n\lg n)$ expected preprocessing time to support retrieving any substring $T[i..i+\ell]$ in $O(\ell+\lg n)$ worst-case time. 
	\end{lemma}
	
	
	\Section{Computing MS within $O(\delta \lg \frac{n}{\delta})$ Words of Space}
	This section presents our most space-efficient solution for computing MS.
	The data structure is similar to LZ-compressed indexes used for computing Longest Common Substrings \cite{gagie2013heaviest}, but instead of using phrases in LZ parse, we use the parsing phrases defined upon the smallest size string attractor of $T[1..n]$: we store one Patricia tree $T_{rev}$ for the reversed parsing phrases; we store another $T_{suf}$ for the suffixes of $T$ starting at different phrase boundaries.
	Those $\gamma$ phrases, $T[i+1.. j]$, and their corresponding suffixes, $T[j+1.. n]$, are sorted in the lexicographic order, respectively.
	We construct a $\gamma \times \gamma$ grid: we add a point, $(x, y)$, on the grid iff the lexicographically $x$-th phrase $T[i+1.. j]$ is followed by the lexicographically $y$-th suffix $T[j+1..n]$ in text $T$; we assign the $x$-coordinate of the point to reversed phrase\footnote{Given a string $s=$``abcd'', the formula $s^{rev}$ represents the string ``dcba''.} $T[i+1..j]^{rev}$ and the $y$-coordinate of the point to the suffix $T[j+1..n]$.
	It follows that all points on the grid are in rank space, i.e.,  they have coordinates on the integer grid $[\gamma]^2=\{1, 2,\cdots, \gamma\}^2$.
	We store a linear space data structure implemented by part ii) of Lemma \ref{lem:induced-check} to support 2D orthogonal range emptiness queries as induced-check and a linear space data structure by part ii) of Lemma \ref{lemma:partner_finding} to support 2D orthogonal range predecessor/successor queries as $\partner$-finding upon the $\gamma$ points on the grid, respectively.
	Finally, we need a data structure supporting substring queries in $T$, implemented by Lemma \ref{lemma: self-index}. 
	Storing the string attractor, both Patricia trees, and the data structures for orthogonal range searching uses $O(\gamma)$ words of space, while the string indexing data structure by Lemma \ref{lemma: self-index} requires $O(\delta \lg \frac{n}{\delta})$ words.
	Overall, the space cost is $O(\delta \lg \frac{n}{\delta})$ words, since $\gamma=O(\delta \lg \frac{n}{\delta})$ \cite{kociumaka2021towards}.

	As the pattern has $m$ entries, $P$ can be partitioned into $m-1$ different prefix and suffix pairs\textemdash that is, $P[1..i]$ and $P[i+1..m]$, for each $1\leq i \leq m-1$.
	For each prefix and suffix pair, we can use Lemma \ref{lemma:loci_finding} to find the loci of $(P[1..i])^{rev}$ in $T_{rev}$ and the loci of $P[i+1..m]$ in $T_{suf}$, respectively.
	
	\begin{lemma}
		\label{lemma:loci_finding}
		Given a pattern $P[1..m]$, for all $1\leq i\leq m-1$, we can find the longest common prefix (LCP) of $(P[1..i])^{rev}$ and the path label of the node where the search in $T_{rev}$ terminates, and the LCP of $P[i + 1..m]$ and the path label of the node where the search in $T_{suf}$ terminates in $O(m^2+m\lg n)$ time.
	\end{lemma}
	\begin{proof}
		For each $1\leq i\leq m-1$, a query need to access $T$ to check that the path labels of the nodes where the searches terminate are really prefixed by some prefixes of $(P[1..i])^{rev}$ and $P[i + 1..m]$, which can be solved by the substring queries using Lemma \ref{lemma: self-index}.
		As there are $m-1$ different pairs of $(P[1..i])^{rev}$ and $P[i + 1..m]$, and the total number of characters that each pair of them contain is $m$, the searching time is $O(m^2+m\lg n)$.
	\end{proof}
	
	Next, we present the query algorithm.
	For $1\leq i\leq m-1$, we search for $(P[1..i])^{rev}$ in $T_{rev}$ and for $P[i+1..m]$ in $T_{suf}$; access $T$ to find the longest common prefix (LCP) of $(P[1..i])^{rev}$  and the path label of the node where the search in $T_{rev}$ terminates, and the LCP of $P[i+1..m]$ and the path label of the node where the search in $T_{suf}$ terminates; take $\loci_1$ and $\loci_2$ to be the loci of those LCPs.
	For each node, $v$, on the path from $\loci_1$ to the root node of $T_{rev}$, we retrieve the lowest ancestor, $u$, of $\loci_2$ in $T_{suf}$ such that $u$ is induced together with $v$, i.e., $u=\partner(v/\loci_2)$.
	Given any tree node $w$, we use $\str(w)$ to denote the path label of $w$ \footnote{If $v$ (resp. $u$) is the loci, then the longest matched prefix of $(P[1..i])^{rev}$ (resp. $P[i+1..m]$) might be a proper substring of the path label of $v$ (resp. $u$). In that case, we let $\str(v)$ (resp. $\str(u)$) denote the longest matched prefix of $(P[1..i])^{rev}$ (resp. $P[i+1..m]$). And $|\str(root)|$ is always 0.}.
	Observe that \footnote{For example, ``abc''.``efg''=``abcefg''.} $(\str(v))^{rev}.\str(u)$ (or $P[i-\len(\str(v))+1..i+\len(\str(u))]$) might be the longest prefix of $P[i-\len(\str(v))+1..m]$ that occurs in $T$; thus, we set $MS[i-\len(\str(v))+1]$ to be $\len(\str(v).\str(u))$ temporarily.
	For each $k\in[1..m-1]$ and $k\leq j \leq m-1$, the longest prefix of $P[k..m]$ might appear somewhere in $T$ crossing the phrase boundary whose immediately left phrase ends with $P[j]$ and immediately right phrase starts with $P[j+1]$; since there are at most $m-k+1$ different types of phrase boundaries, entry $MS[k]$ will finally store the length of the longest prefix of $P[k..m]$ that appears in $T$ after at most $m-k+1$ times of updates.
	\renewcommand{\thefootnote}{\fnsymbol{footnote}}
	The algorithm is shown in ComputingMS\ref{alg:ComputingMS}.\footnotetext{W.l.o.g., we assume that $m>1$. $MS[m]$ is set to 1, if the loci of $P[m]$ on $T_{suf}$ is a non-root node; Otherwise, $MS[m]$ is set to 0.} 
	\begin{figure}[!t]
		\begin{minipage}[t]{0.48\textwidth}\footnotesize
			\begin{algorithm}[H]
				\caption{$^{\ddagger}$ComputingMS1($P[1..m]$, $T_{suf}$, $T_{rev}$)}
				\label{alg:ComputingMS}
				\begin{algorithmic}[1]
					\STATE $MS[1..m] \gets \{0 \cdots 0\}$
					\FOR{$i=1, 2, \ldots, m-1$}
					\STATE Find $\loci_1$ of $(P[1..i])^{rev}$ in $T_{rev}$
					\STATE Find $\loci_2$ of $(P[i+1..m])$ in $T_{suf}$
					\STATE $v \gets \loci_1$
					\WHILE{$v$ is not $\Null$} 
					\STATE $u\gets \partner(v/\loci_2)$
					\STATE $vp\gets \parent(v)$
					\STATE $j \gets \len(\str(v))$
					\WHILE{$j> \len(\str(vp))$}
					\STATE $\ell \gets j+\len(\str(u))$
					\IF{$\ell> MS[i-j+1]$}
					\STATE $MS[i-j+1] \gets \ell$
					\ENDIF
					\STATE $j\gets j-1$
					\ENDWHILE
					\STATE $v \gets vp$
					\ENDWHILE
					\ENDFOR
				\end{algorithmic}
			\end{algorithm}
		\end{minipage}
		\hfill
		\begin{minipage}[t]{0.48\textwidth}
			\begin{algorithm}[H]
				\caption{$^{\ddagger}$ComputingMS2($S[1..m]$)}
				\label{algo_2}
				\begin{algorithmic}[1]
					\STATE $MS[1..m] \gets \{0 \cdots 0\}$
					\STATE $max \gets 0$
					\FOR{$i=1, 2, \ldots, m-1$}
					\IF{$max\leq S[i]$}
					\STATE $max \gets S[i]$
					\ENDIF
					\STATE $MS[i] \gets max$
					\IF{$max>0$}
					\STATE $max \gets max-1$
					\ENDIF
					\ENDFOR
				\end{algorithmic}
			\end{algorithm}
		\end{minipage}
	\end{figure}
				%
\renewcommand*{\thefootnote}{\arabic{footnote}}

We analyze the query time of the algorithm:
As shown in Lemma \ref{lemma:loci_finding}, all locus of LCPs can be found in $O(m^2+m\lg n)$ time; for each $1\leq i\leq m-1$, the while loop at line 6 is operated $O(i)$ times, and all $O(i)$ iterations will call totally $O(i)$ times of $\partner$-finding queries and fill at most $i$ entries of MS; if only $O(\gamma)$ words of space is allowed, each $\partner$-finding query requires $O(\lg^{\epsilon} \gamma)$ time as shown in part ii) of Lemma \ref{lemma:partner_finding}.
The overall query time is $O(m^2+m\lg n+\sum_{i=1}^{m-1} O(i \cdot \lg^{\epsilon} \gamma+i))=O(m^2\lg^{\epsilon} \gamma+ m\lg n)$. The first solution completes.

\Section{Towards Improving the Computing Time}
\label{sect: future}
In this section, we are trying to get rid of factor $\lg^{\epsilon} \gamma$ from term $m^2\lg^{\epsilon} \gamma$ shown in the computing time before.
As a result, the space cost of the new data structure gets worse slightly.
Before showing our new solutions, we introduce a new definition {\em locally potential maximal exact matching (LPMEM)}:


\begin{definition}
	\label{def:LPMEM}
	Given a phrase boundary $k'\in \Gamma$, we refer to a substring $P[i..j]$ as a locally potential maximal exact matching (LPMEM) that crosses the phrase boundary at position $k'$ if substring $P[i..j]$ with an occurrence $T[i'..j']$ such that $i'\leq k'<j'$ holds and the occurrence can neither be extended to the left nor to the right.
\end{definition}

Let $v$ (resp. $u$) denote an ancestor node of $\loci_1$ (resp. $\loci_2$)  in $T_{rev}$ (resp. $T_{suf}$); let $\path(v, \loci_1, T_{rev})$ denote the path between $v$ and $\loci_1$ on $T_{rev}$; let $\path(u, \loci_2, T_{suf})$ denote the path between $u$ and $\loci_2$ on $T_{suf}$. It follows that if i) $v$ and $u$ are induced together; ii) the child of $v$ on $\path(v, \loci_1, T_{rev})$ does not induce with node $u$ in $T_{suf}$; iii) and the child of $u$ on $\path(u, \loci_2, T_{suf})$ does not induce with node $v$ in $T_{rev}$, then $v$ and $u$ together induce a LPMEM, which is $(\str(v))^{rev}.\str(u)$.


%

\paragraph{Basic Properties of LPMEMs.} We discuss the properties of LPMEMs.
Those properties will be useful for designing the data structures and the query algorithm for the second solution.
For simplicity, we call all sub-strings of $P[1..m]$ that appear as LPMEM's in text $T[1..n]$ the LPMEMs of $P$.

\begin{lemma}
	\label{lemma: LPMEMS_to_MS}
	Given a pattern $P[1..m]$ and all $\occ$ LPMEMs of $P$, the matching statistics of $P$ can be computed in $O(\occ+m)$ time.
\end{lemma}
\begin{proof}
	Assume that all $\occ$ LPMEMs have been found, and each LPMEM can be represented by its starting position, $i$, in $P$ and its length, $\ell(i)$.
	Let $S[1..m]$ be an array, in which entry $S[i]$, for each $1\leq i\leq m$, stores $\ell(i)$ if $P[i, i+\ell(i)-1]$ is a LPMEM.
	If there are multiple LPMEMs sharing the same starting position, $i$, in $P$, $S[i]$ stores the largest length.
	It follows that for each $1\leq i\leq m$, $MS[i]$ is equal to $\max(MS[i-1]-1, S[i])$, where $MS[0]$ is 0.
	See ComputingMS\ref{algo_2} for the algorithm.
			%
			%
			%
	Computing array $S[1..m]$ takes $O(\occ+m)$ time and ComputingMS\ref{algo_2} requires $m$ primitive steps; hence, the computation time is $O(\occ+m)$. Note that \footnote{A pattern $P$ with $m$ characters can have at most $\binom{m}{2}$ sub-strings that appear as LPMEM's in text $T[1..n]$.} $\occ\leq m(m-1)/2$.
\end{proof}

We compute the {\em heavy path decomposition} \cite{harel1984fast} of $T_{rev}$ and $T_{suf}$ mentioned before.
For a node $u$ on a heavy path $H$, let $\hproot(u)$ (resp. $\hp\_leaf(u)$) denote the highest (resp. lowest) node of $H$.
We call the highest node of each heavy path {\em light}.
As $T_{rev}$ and $T_{suf}$ each has $\gamma$ leaves, a path from the root to any leaf on $T_{rev}$ or $T_{suf}$ traverses at most $O(\lg \gamma)$ light nodes.
We give a new definition {\em special skyline node list} borrowing the ideas of {\em skyline node list} from \cite{gagie2013heaviest} and  {\em special nodes} from \cite{abedin2018heaviest}.


\begin{definition} 
	\label{def: special_sky}
	For each light node, $w \in T_{suf}$, we identify a set, $\SpecialLeaves(w)$, of leaf nodes in $T_{rev}$  and a set, $\SpecialSkylineList(w)$, of internal nodes in $T_{rev}$  as follows: leaf node $l\in T_{rev}$ is special iff $l$ and $w$ are induced with each other;
	we define special skyline node $v \in \SpecialSkylineList(w)$ if (\romannumeral1) $v$ is a proper ancestor of $\lca(x, y)$ for some special leaves $x$ and $y$, (\romannumeral2) and the child of $v$ on $\path(v, \lca(x, y), T_{rev})$ does not induce with $\partner(v/\hp\_leaf(w))$.
	Following \cite{abedin2018heaviest}, we identify set $\Special(w)$ of nodes in $T_{rev}$, consisting of the special leaves of $w$ and their lowest common ancestors.
\end{definition}

%
As shown before, for each $1\leq i\leq m-1$, we search for $(P[1..i])^{rev}$ in $T_{rev}$ and for $P[i+1..m]$ in $T_{suf}$; take $\loci_1(i)$ and $\loci_2(i)$ to be the locus of those LCPs.
Let $u$ and $v$ denote some ancestors nodes of $\loci_1(i)$ and $\loci_2(i)$ on $T_{rev}$ and $T_{suf}$, respectively; let $\root_1$ (resp. $\root_2$) denote the root node of $T_{rev}$ (resp. $T_{suf}$).
Henceforth, if $u$ and $v$ induce a LPMEM\footnote{To compute the matching statistics, reporting a LPMEM $(\str(u))^{rev}.\str(v)$ verbatim is unnecessary. What we need are its starting position in $P[1..m]$, which is $i-\len(\str(u))+1$, and the length of the LPMEM, which is $\len(\str(u))+\len(\str(v))$.}, $(\str(u))^{rev}.\str(v)$, then $u$ and $v$ are referred to as the the {\em beginning} and the {\em ending} nodes of that LPMEM.
Observe that for the LPMEMs crossing the phrase boundary between $P[i]$ and $P[i+1]$, their beginning nodes stay on $\path(\partner(\loci_2(i)/\loci_1(i)), \loci_1(i), T_{rev})$, and their ending nodes stay on on $\path(\partner(\loci_1(i)/\loci_2(i)), \loci_2(i), T_{suf})$

Let $w_1, \cdots, w_k$ denote a sequence of light nodes on $\path(\root_2, \loci_2(i), T_{suf})$, sorted increasingly by the node depths, such that
$\partner(\loci_1(i)/\loci_2(i))$ is contained in the heavy path rooted by $w_1$, and $w_k$ is the lowest light node above $\loci_2(i)$.
Similarly, let $t_{1}, \cdots t_k$ denote the nodes on $\path(\root_2, \loci_2(i), T_{suf})$ such that $t_k = \loci_2(i)$ and
$t_h = \parent(w_{h+1})$ for $h<k$.
Let $\alpha_f$ be $\partner(t_f/\loci_1(i))$ and $\beta_f$ be $\partner(w_f/\loci_1(i))$, for each $1\leq f \leq k$.


\begin{lemma}
	\label{lem: special_alpha_beta}
	For each $1\leq f \leq k$, $\beta_f$ and $\partner(\beta_f/\loci_2(i))$  induce a LPMEM, and $\alpha_f$  and $\partner(\alpha_f/\loci_2(i))$  induce a LPMEM.
\end{lemma}
\begin{proof}
	Let $u$ denote $\partner(\alpha_f/\loci_2(i))$ in $T_{suf}$.
	Since $\alpha_f$ in $T_{rev}$ and $t_f$ in $T_{suf}$ are induced together, $u$ must be in the sub-tree of $t_f$.
	Suppose that child $c$ of $\alpha_f$ on $\path(\alpha_f, \loci_1(i), T_{rev})$ is induced with $u$.
	As $t_f$ is an ancestor of $u$, $t_f$ in $T_{suf}$ and $c$ in $T_{rev}$ must be induced together, which contradicts with the claim that $\alpha_f$ is $\partner(t_f/\loci_1(i))$.
	Therefore, $c$ can not be induced with $u$ in $T_{suf}$.
	Following the definition of $\partner$, the child of $u$ on $\path(u, \loci_2(i), T_{suf})$ cannot be induced with $\alpha_f$; hence, $u$ and $\alpha_f$ induce a LPMEM.
	The claim that $\partner(\beta_f/\loci_2(i))$ and $\beta_f$ induce a LPMEM follows a similar argument.
\end{proof}

\begin{lemma} 
	\label{lem: u_v_relation_ship}
	Given a node $u$ on $T_{rev}$ and a node $v$ on $T_{suf}$ such that $u$ and $v$ induce a LPMEM, if $v \in \path(w_f, t_f, T_{suf})$ for some $1\leq f\leq k$, then $u\in \path(\alpha_f, \beta_f, T_{rev})$.
\end{lemma}
\begin{proof}
	The proof is similar to the one shown as \cite[Lemma 12]{abedin2018heaviest}. 
	Suppose u is a proper ancestor of $\alpha_f$. 
	Since $\alpha_f$ and $t_f$ are induced together, node $v$, as an ancestor of $t_f$, is also induced with $\alpha_f$.
	Due to this, $u$ and $v$ cannot induce a LPMEM, which generates a contradiction; therefore, $u$ must be in the sub-tree of $\alpha_f$.
	Suppose that u is in the proper sub-tree of $\beta_f$. 
	Since $u$ and $v$ are induced together, and since $v$ is in the sub-tree rooted by $w_f$, $u$ and $w_f$ are induced together, which contradicts with the claim that $\beta_f$ is $\partner(w_f/\loci_1(i))$.
	Therefore, $u$ must be an ancestor of $\beta_f$.
\end{proof}


\begin{lemma} 
	\cite[Lemma 14]{abedin2018heaviest}
	\label{lem: reference}
	For each $1\leq f \leq k$  and any $x \in \path(\alpha_f , \beta_f , T_{rev})/\{\alpha_f \}$, $\partner(x/\loci_2(i))=\partner(x/t_f )=\partner(x/\hp\_leaf(w_f ))$ always holds.
\end{lemma}

For any node $x\in \{\SpecialSkylineList(w_f)\symbol{92}(\alpha_f \cup \beta_f)\}$, it follows that $x$ and $\partner(x/\hp\_leaf(w_f))$ induce a LPMEM because of Lemma \ref{lem: reference} and Definition \ref{def: special_sky}.


\begin{lemma}
	For each $v\in \SpecialSkylineList(w_f)$, it follows that $v \in \Special(w_f)$.
\end{lemma}
\begin{proof}
	Let $x$ and $y$ be a pair of special leaves under $v$ such that child $c$ of $v$ on $\path(v, \lca(x, y), T_{rev})$ does not induce with $\partner(v/\hp\_leaf(w_f))$.
	Since $\partner(v/\hp\_leaf(w_f))$ and $v$ are induced together by some special leaf $\ell$ under the sibling node of $c$, it follows that $\lca(x, \ell)$ or $\lca(y, \ell)$ is $v$, and $v\in \Special(w_f)$.
\end{proof}

\begin{lemma}
	\label{lem: all_LPMEM}
	Given a node $u$ on $\path(\alpha_f, \beta_f, T_{rev})$ and an ancestor node $v$ of $\loci_2(i)$ on $T_{suf}$, if $u$ and $v$ induce a LPMEM, and $u\notin \{\SpecialSkylineList(w_f)\cup \alpha_f \cup \beta_f\}$, then $u=\lca(\ell, \ell')$ for some pair of $\ell, \ell' \in \SpecialLeaves(w_f)$.
\end{lemma}
\begin{proof}
	Since $\beta_f$ and $w_f$ are induced together, and  since $u$ is a proper ancestor of $\beta_f$, there is a special leaf $\ell$ as the common descendant of $\beta_f$ and $u$.
	As $u$ and $v$ induce a LPMEM, and as $u$ is not $\alpha_f$, $v = \partner(u/\loci_2(i))=\partner(u/t_f)$ by Lemma \ref{lem: reference}.
	Since $u$ and $w_f$ are induced together by $\ell$, $v$ must be a descendant of $w_f$.
	There is at least one leaf $\ell'$ under $u$ that makes $u$ and $v$ induced with each other, and $\ell'$ is a special leaf because it is induced with $w_f$.
	If $\ell'$ and $\ell$ are the same, then $\beta_f$ and $v$ are induced together, which contradicts with the claim that $u$ and $v$ induce a LPMEM.
	There are at least two special leaves $\ell$ and $\ell'$ under $u$.
	Since $u \notin \SpecialSkylineList(w_f)$, $u = \lca(\ell, \ell')$. The proof completes.
\end{proof}

\paragraph{The Second Solution.} 
We apply the induced sub-tree defined in \cite[Definition 18]{abedin2018heaviest} to support finding LPMEMs.
An induced sub-tree $T_{rev}(w)$ w.r.t a light node $w\in T_{suf}$ is a tree having exactly $|\Special(w)|$ nodes such that i) each node $l \in \Special(w)$ has a corresponding node $\hat{l}$ in $T_{rev}(w)$; and ii) for each pair of $\ell, \ell' \in \Special(w)$, node $\lca(\ell, \ell')$ in $T_{rev}$ has a corresponding node, as $\lca$ of $\hat{\ell}$ and $\hat{\ell'}$, in $T_{rev}(w)$.
To support finding LPMEMs, we revise the induced sub-trees as follows:
For each internal node $\hat{v}$ of $T_{rev}(w)$, we maintain a pointer $e_0$ pointing to its lowest proper ancestor $\hat{v}'$ (if exists) that belongs to $\SpecialSkylineList(w)$ and a pointer $e_1$ pointing to its corresponding node $v$ in $T_{rev}$;
for each special skyline node $\hat{v}$ of $T_{rev}(w)$, we maintain a pointer $e_2$ pointing to $\partner(v/\hp\_leaf(w))$ in $T_{suf}$.

Since $\sum_{w} |\Special(w)|=O(\gamma \lg \gamma)$ for all light nodes $w\in T_{suf}$, all revised induced sub-trees totally use $O(\gamma\lg \gamma)$ words of space.
Abedin et al. \cite[Lemma 19]{abedin2018heaviest} showed that given a node $\ell \in \Special(w)$, one can find its corresponding node $\hat{\ell}$ in $T_{rev}(w)$ in $O(\lg \lg \gamma)$ time by maintaining an $O(\gamma\lg \gamma)$ word data structure.
In addition, the data structures introduced in the first solution are also required, occupying extra $O(\delta \lg \frac{n}{\delta})$ words of space.
The overall space cost is $O(\delta \lg \frac{n}{\delta})+O(\gamma\lg \gamma)$ words.

We show how to find LPMEMs between $P$ and $T$.
By Lemma \ref{lemma:loci_finding}, we can find locus $\loci_1(i)$ and $\loci_2(i)$ on $T_{rev}$ and $T_{suf}$ in $O(m^2+m\lg n)$ time for all $1\leq i \leq m-1$.
We compute $\partner(\loci_2(i)/\loci_1(i))$ and $\partner(\loci_1(i)/\loci_2(i))$. 
Each $\partner$ operation takes $O(\lg \lg \gamma)$ time by part i) of Lemma \ref{lemma:partner_finding}.
If $\loci_1(i)$ and $\loci_2(i)$ are induced with each other, then there is only one LPMEM crossing the phrase boundary between $P[i]$ and $P[i+1]$, which is $(\str(\loci_1(i)))^{rev}.\str(\loci_2(i))$, and we continue to search for LPMEMs crossing the phrase boundary between $P[i+1]$ and $P[i+2]$.
Otherwise, we iterate through $\path(\root_2, \loci_2(i), T_{suf})$ to find the light nodes $w_1, \cdots, w_k$ and nodes $t_1, \cdots, t_k$ as described before; compute $\alpha_f$ and $\beta_f$ for all $1\leq f \leq k$.
Since $\beta_1=\partner(w_1/\loci_1(i))=\loci_1(i)$, and since $\alpha_k=\partner(t_k/\loci_1(i))=\partner(\loci_2(i)/\loci_1(i))$, each of those LPMEMs has its beginning node on $\path(\beta_1, \alpha_k, T_{rev})$.
We traverse the path from $\beta_1$ to $\alpha_k$.
In general, the beginning nodes on the sub-path from $\beta_f$ to $\alpha_f$ consist of 3 parts: $\alpha_f$, $\beta_f$, and some special nodes between $\alpha_f$ (excluding $\alpha_f$) and $\beta_f$ (excluding $\beta_f$).
Finding LPMEMs induced by $\alpha_f$ and $\partner(\alpha_f/\loci_2(i))$ or $\beta_f$ and $\partner(\beta_f/\loci_2(i))$ is straightforward, taking $O(\lg \lg \gamma)$ time.
There are $2k$ such LPMEMs, and finding all of them takes $O(\lg \gamma \lg \lg \gamma)$ time, since $k=O(\lg \gamma)$.

It remains to find the LPMEMs with their beginning nodes between $\alpha_f$ (excluding $\alpha_f$) and $\beta_f$ (excluding $\beta_f$).
Given an internal tree node $x$, we use $\lMost(x)$ (resp. $\rMost(x)$) to denote the index of the leftmost (resp. rightmost) leaf descendant of $x$.
Since $\beta_f$ is $\partner(w_f/\loci_1(i))$, there exists at least a special leaf of $T_{rev}$ as a descendant of $\beta_f$ that belongs to $\SpecialLeaves(w_f)$, and we use $\ell$ to denote the leftmost one.
Let $\ell'$ denote the rightmost special leaf among the first $\lMost(x)-1$ leaves of $T_{rev}$ and $\ell''$ denote the leftmost special leaf on the right-hand side of the $\rMost(w_f)$-th leaf of $T_{rev}$.
It follows that the lowest special node above $\beta_f$, denoted by $v$, is the lower one between $\lca(\ell', \ell)$ and $\lca(\ell, \ell'')$.
Once $v$ is found, we check whether it is the beginning node of some LPMEM: If the child of $v$ on $\path(v, \loci_1(i), T_{rev})$ does not induce with $\partner(v/\loci_1(i))$, then we report a LPMEM induced by $v$ and $\partner(v/\loci_1(i))$.
Since $v$ is a special node, we find its corresponding node $\hat{v}$ on $T_{rev}(w)$ in $O(\lg \lg \gamma)$ time by \cite[Lemma 19]{abedin2018heaviest}.
Following the pointer $e_0$ stored at $\hat{v}$, we can find the lowest special skyline node, $\hat{v}'$, above $\hat{v}$.
Note that $\hat{v}'$ is the beginning node of some other LPMEM.
We can find that LPMEM following the pointers $e_1$ and $e_2$ stored at $\hat{v}'$ in $O(1)$ time.
We repeat this procedure to iterate over each special skyline node above $\hat{v}$ until  a node whose pointer $e_1$ pointing to $\alpha_f$ is found.
Each of $\ell, \ell', \ell''$ can be found in $O(\lg \lg \gamma)$ time by 2D orthogonal range successor queries\footnote{The 2D orthogonal range successor query is also used for answering the $\partner$ operation.}, e.g., the leaf index of $\ell$ is the $x$-coordinate of the leftmost point in the query range $[\lMost(\beta_f), \rMost(\beta_f)]\times[\lMost(w_f), \rMost(w_f)]$.
After finding the lowest special node $\hat{v}$, reporting the LPMEMs associated with $\SpecialSkylineList(w_f)$ takes $O(\occ(w_f))$ time, where $\occ(w_f)$ denotes the number of reported LPMEMs.
As there are $k$ different such different $\SpecialSkylineList(w_f)$, finding $\occ_i$ LPMEMs between $\loci_1(i)$ and $\loci_2(i)$ requires $O(k\lg \lg \gamma+\occ_i)=O(\lg \gamma \lg \lg \gamma + \occ_i)$ time, since $k=O(\lg \gamma)$.
Considering there are $m-1$ different pairs of locus, $\loci_1(i)$ and $\loci_2(i)$, finding all LPMEMs between $P$ and $T$ takes $O(m\lg \gamma \lg \lg \gamma + \occ)$ time.
The overall query time is $O(m^2+m\lg \gamma \lg \lg \gamma +m\lg n)$, since $\occ=O(m^2)$.
After finding all the LPMEMs, we can use Algorithm \ref{algo_2} to compute the matching statistics.

\Section{Computing MS for a Text Drawn from Constant-Size Alphabet.} 

In the genomic databases, the constant-size texts arise frequently, e.g., the DNA sequence is drawn from $\{A, C, G, T\}$.
When the alphabet size is constant, we can further improve the computation time to be $O(m^2 + m\lg n)$, while maintaining the overall space cost.
In this section, we first give the third solution to the general case such that $T[1..n]$ is drawn from alphabet $[\sigma]$.
In particular, the new solution achieves the improved computation time when $\sigma$ is a constant.
The new solution will apply $\rank$ and $\select$ operations from the succinct data structures.

\begin{lemma}\cite{golynski2006rank}
	\label{lem:sel-rank}
	Let $A[1..n']$ be an array of $n'$ characters drawn from alphabet $[\sigma']$. There exists a data structure constructed upon $A$ using $O(n'\lg \sigma')$ bits of space, supporting $\rank_c(A, i)$ queries in $O(\lg \lg \sigma')$ time and $\select_c(A, i)$ queries in constant time, where $\rank_c(A, i)$ counts the number of character $c$ that appears in $A[1..i]$, and $\select_c(A, i)$ gives the position of the $i$-th occurrence of character $c$ in the sequence.
\end{lemma}

As shown in the second solution, whenever $m=\Omega(\lg \gamma \lg \lg \gamma)$, the query time is bounded by $O(m^2+m\lg n)$; therefore, we only need to consider the case that $m=O(\lg \gamma \lg \lg \gamma)$.
We will modify the data structure used in the first solution applying the technique that solves the ball inheritance problem \cite{chan2011orthogonal}.
%

Before showing the updated data structure, we review the query algorithm in the first solution.
Given a pair of locus $\loci_1(i)$ and $\loci_2(i)$ in $T_{rev}$ and $T_{suf}$ achieved by searching for the longest prefixes of $(P[1..i])^{rev}$ and $P[i+1..m]$ in $T_{rev}$ and $T_{suf}$, respectively, we iterate over each node $v$ on $\path(\loci_1(i), \root_1, T_{rev})$, and compute $\partner(v/\loci_2(i))$ to find the potential longest common prefix between $P[i-\len(\str(v))+1..]$ and $T$.
As shown as \cite[Lemma 10]{abedin2018heaviest}, operation $\partner(v/\loci_2(i))$ can be reduced to a range emptiness query in the range $[\lMost(v), \rMost(v)]\times[\lMost(\loci_2(i)), \rMost(\loci_2(i))]$, finding the $y$-coordinate of the lowest point (a.k.a. a range successor query) within $[\lMost(v), \lMost(v)]\times(\rMost(\loci_2(i), +\infty)$, and finding the $y$-coordinate of the highest point (a.k.a. a range predecessor query) within $[\lMost(v), \rMost(v)]\times(-\infty, \lMost(\loci_2(i)))$.
We observe that: i) as $v$ is changed from $\loci_1(i)$ to $\root_1$, the query range along $y$-axis is fixed;
ii) given any two nodes $s$ and $t$ on $\path(\root_1, \loci_1(i), T_{rev})$, if $s$ is an ancestor of $t$, then $[\lMost(t), \rMost(t)] \subset [\lMost(s), \rMost(s)]$.
These observations can be used to improve the overall query time for multiple $\partner$-finding operations.
We take the predecessor query along $y$-axis within range $[\lMost(v), \rMost(v)]\times(-\infty, \lMost(\loci_2(i)))$ for each node $v$ on $\path(\loci_1(i), \root_1, T_{rev})$ as an example to describe the solution, while the 2D range emptiness queries and 2D range successor queries can be answered similarly.

We number tree levels of $T_{rev}$ incrementally starting from the root
level, which is level 0; refer to the first $\lg^{1+\epsilon} \gamma$ tree levels on the top as {\em active tree levels} for any small constant $\epsilon>0$.
Let $v$ denote any internal node of $T_{rev}$ on some active level;
let $\size(v)$ denote the number of leaves in the sub-tree rooted by $v$; let $\ell(v)$ denote its tree level.
We associate node $v$ with a sequence $S(v)[1..\size(v)]$ storing the coordinates of the points whose $x$-coordinates in the range $[\lMost(v), \rMost(v)]$ and make sure these points are sorted by their $y$-coordinates.
Given sequences $S(v)$, the predecessor query along $y$-axis within $[\lMost(v), \rMost(v)]\times(-\infty, \lMost(\loci_2(i)))$ can be reduced to finding the predecessor of $\lMost(\loci_2(i))$ in one dimension, and any entry $S(v)[j]$, storing the point coordinates, can be accessed in constant time; however, storing all array $S(v)$'s would occupy $O(\gamma \lg^{1+\epsilon} \gamma)$ words of space.
For saving space, we only store $S(\root_1)$ at the root node, but we give a space-efficient data structure that allows to access the point coordinates of any entry, $S(v)[j]$, in constant time, for any node $v$ on active tree levels.

We use the technique that solves the ball inheritance problem in a reversed way.
Let $\tau$ be $\lg^{\epsilon} \gamma$. For simplicity, we assume that both $1/\epsilon$ and $\tau$ are integers.
We assign a color, encoded by some integer, to each active level of $T_{rev}$: Level-0 is colored by $1/\epsilon+1$, Level-$(\lg^{1+\epsilon} \gamma)$ is colored by $0$, while any other Level-$\ell$ is colored by $c(\ell)$, where $c(\ell)=\max$\{$c$ $|$ $(\lg^{1+\epsilon} \gamma - \ell)$ is a multiple of $\tau^c$ and  $0\leq c\leq 1/\epsilon+1$\}.
At each internal node $v$ on active tree levels, we store $(\tau-1)\cdot c(\ell(v))$ arrays of {\em skipping pointers}, denoted by $\sp$.
For each $0\le t\le c(\ell(v))-1$ and $1\le k \le \tau-1$, array $\sp(v, t, k)$ has the same number of entries as of $S(v)$; if point $S(v)[j]$ is stored in any array $S(\cdot)$ associated with nodes at level $\ell(v)+\tau^t\cdot k$, then the $j$-th entry of array $\sp(v, t, k)$ stores the descendant, denoted by $\desc(v, t, k, j)$, of $v$ at level $\ell(v)+\tau^t\cdot k$ containing the point $S(v)[j]$, and the descendant is encoded by its rank among all the descendants of $v$ at the level $\ell(v)+\tau^t\cdot k$ in the left-to-right order; otherwise, entry $\sp(v, t, k)[j]$ is set to be $-1$.
We use Lemma \ref{lem:sel-rank} to support $O(1)$-time $\select$ over array $\sp(v,t, k)$.
Recall that within array $S(v)$ and $S(\desc(v, t, k, j))$, points are ordered by their $y$-coordinates; therefore, a $\select_{\sp(v, t, k)[j]}(\sp(v, t, k), j')$ query returns the array index of point $S(\desc(v, t, k, j))[j']$ in array $S(v)$, for each $1\le j' \le |S(\desc(v, t, k, j))|$. 
In general, to retrieve the point coordinates of any entry, $S(v)[j]$, we find the lowest ancestor $v'$ of $v$ such that $c(\ell(v'))>c(\ell(v))$, and then use the query $\select_{r(v)}(\sp(v', c(\ell(v)), \frac{\ell(v)-\ell(v')}{\tau^{c(\ell(v))}}), j)$ to locate the array index of point $S(v)[j]$ in array $S(v')$, where $r(v)$ denotes the rank of $v$ among all the descendants of $v'$ at the level $\ell(v)$. 
One hop\footnote{Since there is unique ancestor $v'$ of $v$ that $v$ can hop over to, we simply store a pointer that pointing to $v'$ and the rank $r(v)$ w.r.t. $v'$ at node $v$ in the preprocessing stage.} from $v$ to $v'$ increases the node color by at least one.
Therefore, after at most $1/\epsilon+1 - c(\ell(v))$ hops, we reach the root level, where we can immediately retrieve the point coordinates stored in $S(\root_1)$.
Since each $\select$ query takes constant time, the overall query time is $O(1)$.

We analyze the space cost for storing all $\sp(v, t, k)$'s. 
As the size of the alphabet that $T[1..n]$ is drawn from is $\sigma$, the out-degree of each node in $T_{rev}$ is at most $\sigma$.
There are at most $\min(\sigma^{\tau^{t}\cdot k}, \gamma)$ descendants of $v$ on the tree level $\ell(v)+ \tau^{t}\cdot k$, and the rank of each of them can be encoded within $\lg (\sigma^{\tau^{t}\cdot k})$ bits of space.
Clearly, there are $\frac{\lg^{1+\epsilon} \gamma}{\tau^c}$ active levels colored in $c$ for each $0\le c \le 1/\epsilon+1$.
Fix $t$ and $k$, and the total number of entries in $\sp(v,t, k)$ for all nodes at the same tree level is at most $\gamma$.
The overall space cost in bits is at most,
$$\sum_{c=0}^{1/\epsilon+1} (\frac{\lg^{\epsilon+1} \gamma}{\tau^c} \sum_{t=0}^{c-1} (\sum_{k=1}^{\tau-1} (\gamma\lg (\sigma^{\tau^{t}\cdot k}))))\leq \gamma\lg \sigma\sum_{c=0}^{1/\epsilon+1} (\frac{\lg^{\epsilon+1} \gamma}{\tau^c} \tau^{c-1}\sum_{k=1}^{\tau-1} k)=O(\gamma \lg \sigma\lg^{2\epsilon+1} \gamma).$$
Note that the data structure supporting $\select$ queries upon $\sp(v, t, k)$ has the same space upper bound as the one for storing $\sp(v, t, k)$.
Hence, Lemma \ref{lem:ball_rev} follows.

\begin{lemma}
	\label{lem:ball_rev}
	We can build a data structure of $O(\gamma \lg \sigma\lg^{2\epsilon+1} \gamma)$ bits of space upon $T_{rev}$ such that later, given a node $v$ on any active level, one can find the point coordinates of any entry of $S(v)$ in constant time.
\end{lemma}

Next, we show the data structure that can support computing multiple $\partner$-finding operations efficiently:
We construct a sequence $R(v)[1..\size(v)]$ at each internal node $v$ on active tree levels such that if the point $S(v)[j]$ is stored in $S(v_s)$ in the next level, where $v_s$ denotes the $s$-th child of $v$ in the left-to-right order, then $R(v)[j]$ is set to be $s$; use Lemma \ref{lem:sel-rank} to support $O(\lg \lg \sigma)$-time $\rank$ over $R(v)$.
Since the out-degree of each node in $T_{rev}$ is at most $\sigma$, $R(v)$ and its associated data structure occupy $O(\size(v)\lg \sigma)$ bits of space.
All $R(v)$'s stored on active tree levels use $O(\gamma \lg \sigma \lg^{1+\epsilon} \gamma)$ bits of space.
Finally, we use Lemma \ref{lem:ball_rev} to access the coordinates of any entry $S(v)[j]$, which occupies $O(\gamma \lg \sigma\lg^{2\epsilon+1} \gamma)$ bits of space additionally.

We describe how to find the predecessor along $y$-axis within $[\lMost(v), \rMost(v)]\times(-\infty, \lMost(\loci_2(i)))$ for each node $v$ on $\path(\loci_1(i), \root_1, T_{rev})$.
We traverse through $\path(\loci_1(i), \root_1, T_{rev})$ reversely, i.e., from the root node to $\loci_1(i)$.
At the root node, the index, $j_{\root_1}$, of the proper predecessor of $\lMost(\loci_2(i))$ in $S(\root_1)$ is $\lMost(\loci_2(i))-1$, because $S(\root_1)$ contains all the coordinates of $\gamma$ points in rank space and those points are increasingly sorted by $y$-coordinates.
We immediately return the $y$-coordinate of $S(\root_1)[j_{\root_1}]$ in constant time.
In general, given two nodes $u$ and $v$ on $\path(\loci_1(i), \root_1, T_{rev})$ such that $v$ is the $e$-th child of $u$ in the left-to-right order, if we know the index, $j_u$, of the predecessor of $\lMost(\loci_2(i))$ in $S(u)$, then the index, $j_v$, of the predecessor of $\lMost(\loci_2(i))$ in $S(v)$ can be located in $O(\lg \lg \sigma)$ time by $\rank_e(R(u), j_u)$; and then we use Lemma \ref{lem:ball_rev} to access the $y$-coordinate of $S(v)[j_v]$ in constant time.
As there are $i$ characters in $(P[1..i])^{rev}$, there are at most $i$ nodes on $\path(\loci_1(i), \root_1, T_{rev})$.
Hence, all predecessor queries along the path can be answered in $O(i \lg \lg \sigma)$ time.
As a result, the $\partner(v/\loci_2(i))$ queries for all $v$ on $\path(\loci_1(i), \root_1, T_{rev})$ can be answered in $O(i \lg \lg \sigma)$ time.
Considering there are $m-1$ different pairs of $\loci_1(i)$ and $\loci_2(i)$, the query time of computing MS is $\sum_{i=1}^{m-1}O(i\lg \lg \sigma)=O(m^2\lg \lg \sigma)$ time, given that all pairs of locus are available.
Since finding $m-1$ pairs of locus requires $O(m^2+m\lg n)$ time by Lemma \ref{lemma:loci_finding}, the overall query time is $O(m^2 \lg \lg \sigma+m\lg n)$.

In the beginning of this section, we assume that $m$, the length of the query pattern, is bounded by $O(\lg \gamma \lg \lg \gamma)$.
As mentioned before, once $m$ is $\Omega(\lg \gamma \lg \lg \gamma)$, we can apply the solution shown in part (ii) of Theorem \ref{theorem_for_all} to compute the matching statistics in $O(m^2+m\lg n)$ time with an $O(\gamma \lg \gamma + \delta \lg \frac{n}{\delta})$ word data structure. 
Combining both solutions yields Theorem \ref{theorem_gamma: MS_2}.

\begin{theorem}
	\label{theorem_gamma: MS_2}
	Given a text $T[1..n]$ drawn from $[\sigma]$, we can build a data structure for $T[1..n]$ with $O(\gamma \lg \gamma + \delta \lg \frac{n}{\delta} + \frac{\gamma}{\log_\sigma n}\lg^{2\epsilon+1} \gamma)$ words of space, for any small constant $\epsilon>0$, such that later, given a pattern $P[1..m]$, we can compute $MS$ for P w.r.t. $T$ in $O(m^2 \lg \lg \sigma + m\lg n)$ time, assuming that the number of bits in a word is $\Omega(\lg n)$.
\end{theorem}


\begin{corollary}
	\label{coro: constant-sigma}
	Given a text $T[1..n]$ drawn from constant-size alphabet, we can build a data structure for $T[1..n]$ with $O(\gamma \lg \gamma + \delta \lg \frac{n}{\delta})$ words of space, such that later, given $P[1..m]$, we can compute $MS$ for P w.r.t $T$ in $O(m^2 + m\lg n)$ time.
\end{corollary}

\paragraph{Acknowledgments.}The author would like to thank Travis Gagie and Meng He for discussing various topics related to the compact data structures, and especially thank Travis for sharing this research topic as a course project. 
The author would also like to thank the anonymous reviewers for their valuable comments and suggestions.

\Section{References}
\bibliographystyle{IEEEbib}
\bibliography{refs}


\end{document}